\documentclass[11pt]{article}

\usepackage[margin=1in]{geometry}
\usepackage{amsmath,amssymb,amsthm,mathtools}
\usepackage{booktabs}
\usepackage{tikz}
\usetikzlibrary{arrows.meta}
\usepackage{enumitem}
\usepackage{placeins}
\usepackage[linesnumbered,ruled,vlined]{algorithm2e}
\usepackage{microtype}
\usepackage[round]{natbib}
\usepackage[hidelinks]{hyperref}
\usepackage[capitalise]{cleveref}

\newtheorem{theorem}{Theorem}
\newtheorem{lemma}{Lemma}
\newtheorem{proposition}{Proposition}
\newtheorem{corollary}{Corollary}

\theoremstyle{definition}

\newtheorem{example}{Example}

\newcommand{\cl}{\operatorname{cl}}
\newcommand{\B}{\mathcal B}
\newcommand{\cF}{\mathcal F}
\SetKw{Return}{return}

\usepackage{authblk}

\title{Top Trading Cycles with Indifferences under Matroid Constraints}
\author{Yasushi Kawase}
\affil{Chuo University}
\date{}
\hypersetup{
  pdftitle={Top Trading Cycles with Indifferences under Matroid Constraints},
  pdfauthor={Yasushi Kawase}
}

\begin{document}
\maketitle

\begin{abstract}
We study the reallocation of indivisible objects among agents who may hold initial endowments and may be indifferent between objects.
Each agent has her own set of admissible objects, and a matroid constraint specifies which combinations of assigned objects are feasible.
The model contains the housing market, house allocation with newcomers who hold no endowment, and the assignment of students to schools under matroidal constraints at each school.
For this model, we introduce the top-class trading cycles mechanism (TCTC), which extends the top trading cycles mechanism (TTC) to matroid constraints and weak preferences.
When a trading cycle is selected, TCTC fixes each agent's top indifference class rather than a specific object and postpones the choice within the class to the end.
TCTC is individually rational, Pareto efficient, and weakly group-strategy-proof on the full domain of weak preferences, and it runs in polynomial time.
When preferences and reports are strict, TCTC is strongly group-strategy-proof.
Together with a known impossibility result, our theorem establishes a maximal domain of school-by-school constraints for individual rationality, Pareto efficiency, and strategy-proofness.
\end{abstract}

\section{Introduction}

The reallocation of indivisible objects among agents who may hold initial endowments is a classical problem in market design.
In the housing market of \citet{ShapleyScarf1974}, each agent owns one object, and the top trading cycles (TTC) mechanism is individually rational (IR), Pareto efficient (PE), and group-strategy-proof.
Applications differ from this benchmark in two respects.
Some agents, such as newcomers in house allocation~\citep{AbdulkadirogluSonmez1999}, hold no initial endowment.
Feasibility is subject to constraints beyond a capacity per object, such as type-specific or regional quotas in school assignment~\citep{AbdulkadirogluSonmez2003,KamadaKojima2015} and distributional requirements in teacher reassignment~\citep{CombeTercieuxTerrier2022}.

Both features are captured by the following framework.
Each agent has her own set of admissible objects, such as the seats she could take at schools or the option of remaining unmatched, and she may initially hold one of them.
An allocation selects one object for every agent, and a constraint specifies which combinations of selected objects are feasible, possibly coupling the choices of different agents.
An allocation is feasible when its assigned objects form a base of a given matroid.
The matroid exchange property provides feasible replacements of individual objects.
To perform several such replacements simultaneously, TTC-type mechanisms must also coordinate the choice of exchanges.
Under strict preferences, treating each object as a school embeds this framework in the M-convex model of \citet{SuzukiEtAl2023}, whose TTC-M mechanism guarantees IR, PE, and strong group-strategy-proofness.
For school-by-school constraints, \citet{ImamuraKawase2025} show that generalized matroids (g-matroids) form a maximal class of school constraints for which IR, PE, and strategy-proofness (SP) can be guaranteed.
These mechanisms, like TTC itself, rely on strict preferences.

Indifferences are common in applications, and in the housing market several TTC-type mechanisms handle them~\citep{AlcaldeUnzuMolis2011,JaramilloManjunath2012}.
To run a TTC-type mechanism under indifferences, one has to specify how an agent chooses among her tied objects when she points.
A naive choice rule can violate Pareto efficiency, and a rule designed to restore it can violate strategy-proofness.
We postpone the choice among tied objects and obtain IR, PE, and weak group-strategy-proofness (WGSP) on the full domain of weak preferences under matroid constraints.

\subsection{Our results}

We introduce the \emph{top-class trading cycles mechanism} (TCTC), which fixes each agent's top indifference class when she trades and selects the objects only at the end.
TCTC is IR, PE, and WGSP on the full domain of weak preferences under matroid constraints, and it runs in polynomial time given a membership oracle for the constraint (\Cref{thm:base-market}).
WGSP rules out every joint misreport that makes all coalition members strictly better off.
On the full domain of weak preferences, strong group-strategy-proofness is incompatible with PE even in the housing market~\citep{Ahmad2021}, and WGSP is therefore a natural group-incentive requirement.

When preferences and reports are strict, TCTC is an instance of the TTC-M mechanism of \citet{SuzukiEtAl2023}: each object is viewed as a school, and the priorities are chosen to match those of TCTC.
It therefore inherits strong group-strategy-proofness in every base market (\Cref{cor:strict-gsp}).
The impossibility result of \citet{ImamuraKawase2025} for school-by-school constraints holds already for strict preferences and thus also applies to the full domain of weak preferences.
Together with TCTC, it establishes that g-matroids remain a maximal domain of school constraints for IR, PE, and SP on the full domain of weak preferences (\Cref{cor:maximality}).
TCTC also rules out feasible coalitional improvements in which every member is strictly better off and the coalition's assigned objects span the same flat as its initial endowments (\Cref{cor:core}).
This property reduces to the usual core condition in housing markets.

\subsection{Related work}

TCTC belongs to the line of TTC mechanisms that starts with the housing market of \citet{ShapleyScarf1974}.
\citet{AbdulkadirogluSonmez1999} study house allocation with existing tenants and newcomers through a priority-based TTC algorithm.
In our base market, a newcomer is initially endowed with an agent-specific unmatched object, and under strict preferences and unit capacities TCTC specializes to their mechanism.

Indifferences led to several extensions of TTC in the housing market, including top trading absorbing sets~\citep{AlcaldeUnzuMolis2011}, top cycle rules~\citep{JaramilloManjunath2012}, and their unifications~\citep{AzizDeKeijzer2012,SabanSethuraman2013}.
Fixed exogenous tie-breaking followed by TTC is strategy-proof but in general guarantees only weak Pareto efficiency~\citep{Ehlers2014}.
\citet{Ahmad2021} provides a sufficient condition for WGSP and applies it to several extensions of TTC to weak preferences.
In the housing-market special case, TCTC recovers the IR, PE, and WGSP guarantees established in this literature.

A second line of work extends TTC to constrained markets with strict preferences.
\citet{SuzukiEtAl2023} develop TTC-M for M-convex distributional constraints on the numbers of students assigned to schools~\citep{Murota2003,MurotaShioura1999}.
\citet{ZhangTangYin2023} study matroid constraints on the set of holders of each object.
\citet{ImamuraKawase2025} impose the constraint on the set of matched student--school pairs, allow unmatched students, and formulate generalized TTC for g-matroid constraints, which contain both of the preceding classes.
These constraints may couple assignments across schools.
They obtain IR, PE, and strong group-strategy-proofness by reducing generalized TTC to TTC-M in a virtual market with M-convex constraints.
\citet{ZhangTangMiaoYin2024} propose a local-exchange mechanism for matroid constraints and weak preferences.\footnote{Contrary to the feasibility claim in their Theorem~1, their Algorithm~1 maps the initial allocation $(a,a,b)$ to $(a,b,a)$ when $a\succ_1 b$, $b\succ_2 a$, and $a\succ_3 b$, and the sole constraint is that agents~1 and~3 cannot both receive~$a$.}

\section{Preliminaries}\label{sec:preliminaries}

Let $U$ be a finite ground set.
For $X\subseteq U$ and $x,y\in U$, we write $X-x$ for $X\setminus\{x\}$ and $X+y$ for $X\cup\{y\}$.
A nonempty family $\B\subseteq 2^U$ is a \emph{matroid base family} if it satisfies the basis-exchange axiom: for any $B,B'\in\B$ and $x\in B\setminus B'$, there is a $y\in B'\setminus B$ such that $B-x+y\in\B$.
The members of $\B$ are called bases, and they all have the same number of elements.
This common number is called the \emph{rank} of $\B$.
A set is \emph{independent} if it is contained in a base, and an element of $U$ is a \emph{loop} if it belongs to no base.

\subsection{Base markets}\label{subsec:base-market}

A \emph{base market} is a tuple $(I,(E_i)_{i\in I},(\succsim_i)_{i\in I},\B,\omega)$.
The set $I=\{1,2,\ldots,n\}$ is a finite set of agents.
For each agent $i$, the finite set $E_i$ contains the \emph{objects} that can be assigned to agent $i$ and to no one else.
The sets $E_i$ are pairwise disjoint, and we write $E\coloneqq\bigcup_{i\in I}E_i$.
Each agent $i$ has a complete and transitive preference relation $\succsim_i$ on $E_i$.
We refer to such relations as weak preferences and write $\succ_i$ and $\sim_i$ for the strict preference and indifference relations induced by $\succsim_i$, respectively.
The family $\B$ is a rank-$n$ matroid base family on $E$.
We assume that $\B$ is accessible through a \emph{membership oracle}, which decides whether a given set of $n$ objects belongs to $\B$.
An \emph{allocation} is a map $\mu\colon I\to E$ with $\mu(i)\in E_i$ for every agent $i$, which we identify with the set $\{\mu(i)\mid i\in I\}$.
Thus, an allocation is a subset of $E$ that contains exactly one object from every $E_i$, and it is \emph{feasible} if it belongs to $\B$.
A base of $\B$ may contain more than one object from one set $E_i$ and none from another set $E_j$.
Such bases are not allocations, and the mechanism uses them only to test exchanges.
The \emph{initial endowment} $\omega$ is an allocation, and we assume that it is feasible.
For every $J\subseteq I$, write $\bar{J}\coloneqq I\setminus J$, $E_J\coloneqq\bigcup_{i\in J}E_i$, and $\omega(J)\coloneqq\{\omega(i)\mid i\in J\}$.

A direct mechanism $\Phi$ assigns a feasible allocation $\Phi(\succsim_I)$ to every preference profile $\succsim_I=(\succsim_i)_{i\in I}$, and $\Phi_i(\succsim_I)$ denotes the object assigned to agent $i$.
The mechanism is \emph{individually rational} (IR) if $\Phi_i(\succsim_I)\succsim_i\omega(i)$ for every $\succsim_I$ and every $i$.
It is \emph{Pareto efficient} (PE) if, for every preference profile $\succsim_I$, no feasible allocation $\mu$ satisfies $\mu(i)\succsim_i\Phi_i(\succsim_I)$ for every $i$ and $\mu(k)\succ_k\Phi_k(\succsim_I)$ for some $k$.
It is \emph{strategy-proof} (SP) if there are no profile $\succsim_I$, agent $i$, and report $\succsim'_i$ such that $\Phi_i(\succsim'_i,\succsim_{-i})\succ_i\Phi_i(\succsim_I)$.
It is \emph{weakly group-strategy-proof} (WGSP) if there are no profile $\succsim_I$, nonempty coalition $Q\subseteq I$, and joint report $\succsim'_Q$ such that $\Phi_i(\succsim'_Q,\succsim_{-Q})\succ_i\Phi_i(\succsim_I)$ for every $i\in Q$.
Thus, WGSP requires at least one coalition member not to improve strictly.
It is \emph{strongly group-strategy-proof} (strong GSP) if no coalition has a joint misreport that makes every member weakly better off and at least one member strictly better off.
Strong GSP implies WGSP, which implies SP.

\subsection{Student--school markets}\label{subsec:model}

An important special case is the assignment of students to schools, in which agents are students and objects are student--school pairs.
Let $S$ be a finite set of schools, and let $E_i\coloneqq\{i\}\times(S\cup\{\varnothing\})$ for every student $i$.
Thus, $E=I\times(S\cup\{\varnothing\})$.
The object $(i,s)$ assigns student $i$ to school $s$, and $(i,\varnothing)$ leaves her unmatched.
Thus, a weak preference on $E_i$ is a weak preference over $S\cup\{\varnothing\}$, and an allocation assigns each student to at most one school.
For $B\subseteq E$ and $s\in S$, let $B(s)\coloneqq\{i\in I\mid (i,s)\in B\}$ denote the set of students that $B$ assigns to $s$.

Each school $s$ has a nonempty family $\cF_s\subseteq 2^I$ of sets of students that it can accept.
Define
\begin{align*}
\B\coloneqq\left\{B\subseteq E\ \middle|\ |B|=|I|,\ B(s)\in\cF_s\ (\forall s\in S)\right\}.
\end{align*}
An allocation is therefore feasible if and only if each school $s$ can accept the set of students assigned to it.
A nonempty family $\cF\subseteq 2^U$ on a finite ground set $U$ is a \emph{generalized matroid}, or \emph{g-matroid}, if for any $X,Y\in\cF$ and any $e\in X\setminus Y$, at least one of the following holds~\citep{Tardos1985}:
\begin{enumerate}[label=\textup{(\roman*)},leftmargin=2.5em]
  \item $X-e\in\cF$ and $Y+e\in\cF$;
  \item there is an $e'\in Y\setminus X$ such that $X-e+e'\in\cF$ and $Y-e'+e\in\cF$.
\end{enumerate}
The direct sum of g-matroids on disjoint ground sets is a g-matroid, and a nonempty cardinality truncation $\{X\in\cF\mid |X|=k\}$ of a g-matroid is a matroid base family~\citep{Tardos1985}.
Suppose that every $\cF_s$ is a g-matroid.
View each $\cF_s$ as a g-matroid on the block $I\times\{s\}$ by replacing a set $X\subseteq I$ with $\{(i,s)\mid i\in X\}$, and place the g-matroid $\{\emptyset,\{(i,\varnothing)\}\}$ on each singleton $\{(i,\varnothing)\}$.
The direct sum of these g-matroids is $\{B\subseteq E\mid B(s)\in\cF_s\ (\forall s\in S)\}$, and $\B$ is its cardinality-$|I|$ truncation, which is nonempty because it contains $\omega$.
Hence, $\B$ is a matroid base family, and the student--school market is a base market.
Unit capacities give the housing market with weak preferences, and students with $\omega(i)=(i,\varnothing)$ are newcomers without a house, as in \citet{AbdulkadirogluSonmez1999}.

\section{The mechanism}\label{sec:mechanism}

We describe the mechanism for a base market $(I,(E_i)_{i\in I},(\succsim_i)_{i\in I},\B,\omega)$.
The mechanism uses a fixed priority order over agents, and we index the agents so that smaller indices have higher priority.
Fix a report-independent strict total order $\sqsubset$ on $E$.
For distinct bases $B,B'\in\B$, write $B\mathrel{\triangleleft}B'$ if the $\sqsubset$-minimum object in $B\mathbin{\triangle}B'$ belongs to $B$.
This defines a strict total order on $\B$, and it is used only to make the final allocation single-valued.

\subsection{Residual base family}\label{subsec:residual}

The mechanism processes agents in iterations.
When it processes an agent $f$, it fixes an indifference class $T_f\subseteq E_f$ but does not yet choose an object from that class.
Suppose that it has processed a subset $F\subseteq I$ of agents in this way.
Because the sets $E_i$ are pairwise disjoint, objects selected from the classes $(T_f)_{f\in F}$ are automatically distinct and do not belong to $E_{\bar{F}}$.
The \emph{residual base family} consists of the size-$|\bar{F}|$ subsets of $E_{\bar{F}}$ that can be completed to a base of $\B$ by selecting one object from each fixed class:
\begin{equation}
 \tilde{\B}\coloneqq\left\{A\subseteq E_{\bar{F}}\ \middle|\ |A|=|\bar{F}|,\ A\cup\{y_f\mid f\in F\}\in\B\text{ for some }(y_f)_{f\in F}\in\prod_{f\in F}T_f\right\}.
 \label{eq:residual}
\end{equation}
The family $\tilde{\B}$ is an auxiliary base family, and a member of $\tilde{\B}$ need not contain one object from every $E_i$ with $i\in\bar{F}$.
The mechanism uses this family only to test matroid exchanges.
The following lemma justifies applying matroid exchanges to the residual base family.

\begin{lemma}\label{lem:residual}
If $\tilde{\B}$ is nonempty, then it is a rank-$|\bar{F}|$ matroid base family on $E_{\bar{F}}$.
Moreover, every nonloop of $\tilde{\B}$ is a nonloop of $\B$.
\end{lemma}

\begin{proof}
We prove the first assertion by induction on $|F|$.
If $F=\emptyset$, then $\tilde{\B}=\B$.
Otherwise, fix $f\in F$ and let $\B'$ be the residual base family for $F-f$ on $E_{\bar F}\cup E_f$.
This family is nonempty because any completion witnessing $\tilde{\B}\ne\emptyset$ also yields a member of $\B'$.
Assume, as the induction hypothesis, that $\B'$ is a matroid base family.
For $A\subseteq E_{\bar F}$, \eqref{eq:residual} gives $A\in\tilde{\B}$ if and only if $A+t\in\B'$ for some $t\in T_f$.

Take $A,A'\in\tilde{\B}$ and $x\in A\setminus A'$, and choose $t,t'\in T_f$ with $A+t,A'+t'\in\B'$.
By basis exchange, $A+t-x+y\in\B'$ for some $y\in(A'+t')\setminus(A+t)$.
If $y\in A'\setminus A$, then $A-x+y\in\tilde{\B}$.
Otherwise, $y=t'\ne t$, and $A-x+t+t'\in\B'$.
Exchange again with $A'+t'$, now removing $t$.
Since $t'$ belongs to both bases and $x\notin A'$, the inserted object lies in $A'\setminus A$.
Hence, $A-x+z+t'\in\B'$ for some $z\in A'\setminus A$, giving $A-x+z\in\tilde{\B}$.
This proves basis exchange for $\tilde{\B}$, whose members all have size $|\bar F|$.

Finally, every member of $\tilde{\B}$ extends to a base of $\B$, which gives the nonloop claim.
\end{proof}

\subsection{Pointing and cycle processing}\label{subsec:rule}

We now describe one iteration of the mechanism.
In each iteration, every remaining agent determines her top class and points to a remaining agent.
The mechanism then selects a cycle of the resulting pointing graph and fixes the top classes of the agents on the cycle.
After the last iteration, it selects the final allocation from the fixed classes.
A \emph{state} of the mechanism consists of the set $F$ of processed agents and their fixed classes $(T_f)_{f\in F}$, and it determines the residual base family $\tilde{\B}$ through \eqref{eq:residual}.
Initially, $F=\emptyset$ and $\tilde{\B}=\B$.
The mechanism updates the state so that the initial endowments of the remaining agents always form a base of the residual base family, that is, $\omega(\bar{F})\in\tilde{\B}$.

Each remaining agent points to a remaining agent, possibly herself.
The nonloops of $\tilde{\B}$ are precisely the objects contained in at least one base of $\tilde{\B}$.
For each remaining agent $i\in\bar{F}$, let $T_i$ be her most preferred indifference class in $E_i$ that contains a nonloop of $\tilde{\B}$.
We call $T_i$ agent $i$'s \emph{current top class}.
This class exists because $\omega(i)$ belongs to the base $\omega(\bar{F})$ of $\tilde{\B}$ and is therefore a nonloop.
Define the pointing map $\tau\colon\bar{F}\to\bar{F}$ by
\begin{equation}
  \tau(i)\coloneqq\min\left\{j\in\bar{F}\ \middle|\ \omega(\bar{F})-\omega(j)+a\in\tilde{\B}\text{ for some }a\in T_i\right\},
  \label{eq:target}
\end{equation}
where the minimum is taken with respect to the index priority.
The set in \eqref{eq:target} is nonempty because $T_i$ contains a nonloop and $\omega(\bar{F})$ is a base of $\tilde{\B}$.
Thus, agent $i$ points to the highest-priority remaining agent whose initial endowment can be replaced by some object in her top class, and she does not select a particular object from that class.
All agents choose their arrows simultaneously from the same state.
An arrow entering $i$ indicates that $\omega(i)$ is replaced if the selected cycle contains $i$.

The \emph{pointing graph} is the directed graph on $\bar{F}$ with the arrows $(i,\tau(i))$ for $i\in\bar{F}$.
Every remaining agent has one outgoing arrow, and hence the pointing graph contains a directed cycle.
The mechanism selects any such cycle $C$ and processes it.
Each agent in $C$ is the target of exactly one arrow from an agent in $C$, and hence the initial endowment of every agent in $C$ is used exactly once.
\emph{Processing the cycle $C$} means fixing the current class of each agent in $C$, adding the agents in $C$ to $F$, and redefining the residual base family by \eqref{eq:residual}.
No particular object is selected from a fixed class at this stage.
After the last iteration, the mechanism selects a base of $\B$ that contains an object from every fixed class.
Because $T_i\subseteq E_i$ for every $i$, such a base contains exactly one object from every $E_i$ and is therefore a feasible allocation.

We call the following mechanism the \emph{top-class trading cycles mechanism} (TCTC), and write $\Phi$ for the direct mechanism that it defines.
\Cref{alg:tctc} gives its complete description.

\begin{algorithm}[htbp]
\caption{Top-class trading cycles (TCTC)}\label{alg:tctc}
\small
$F\gets\emptyset$ and $\tilde{\B}\gets\B$\;
\While{$F\ne I$}{
  For each $i\in\bar F$, compute her most preferred indifference class $T_i$ containing a nonloop of $\tilde\B$\;
  Compute $\tau(i)$ for every $i\in\bar F$ according to \eqref{eq:target}\;
  Choose any directed cycle $C$ of the pointing graph at the current state\;
  Fix $T_i$ for every $i\in C$ without selecting an object for any agent\;
  $F\gets F\cup C$\;
  Redefine $\tilde{\B}$ from $\B$ and $(T_f)_{f\in F}$ according to \eqref{eq:residual}\;
}
\Return the minimum $B\in\B$ with respect to $\mathrel{\triangleleft}$ that contains an object from every class $T_i$\;
\end{algorithm}
\FloatBarrier

In the next section, we show that TCTC produces a feasible allocation that does not depend on which cycle is selected at each step, and that it is IR, PE, and WGSP.
Before turning to the analysis, we illustrate the mechanism on an instance in which an early choice within a class can lead to Pareto inefficiency.

\begin{example}\label{ex:school}
Consider two students $I=\{1,2\}$ and two schools $a$ and $b$ with capacities two and one, respectively.
Thus, $\cF_a=2^I$ and $\cF_b=\{X\subseteq I\mid |X|\le1\}$.
Both students initially attend $a$, giving $\omega=(a,a)$ and leaving $b$ initially empty.
Their preferences are $a\sim_1 b\succ_1\varnothing$ and $b\succ_2 a\succ_2\varnothing$.
The common priority order is $1,2$.

Initially, both students point to student 1.
Student 1 points to herself because her initial object $(1,a)$ is in her top class.
Student 2 points to student 1 because replacing $(1,a)$ in the initial base by $(2,b)$ gives the base $\{(2,a),(2,b)\}$.
This base is used only to test the exchange; it is not an allocation.
Processing student 1's self-cycle fixes $T_1=\{(1,a),(1,b)\}$ without selecting either school.

Student 2 then points to herself.
Replacing her initial object $(2,a)$ by $(2,b)$ is feasible in the residual base family because $(2,b)$ can be combined with $(1,a)\in T_1$ to form a base.
The mechanism fixes $T_2=\{(2,b)\}$.
Since $b$ has capacity one, the unique feasible allocation consistent with both fixed classes is $\mu=(a,b)$.
Thus, TCTC assigns student 2 to the initially empty school $b$, changing the school sizes from $(2,0)$ to $(1,1)$ (\Cref{fig:tctc-example}).

Fixing student 1's class does not commit her to either school.
Assigning her to $b$ at that stage would leave school $a$ as student 2's best remaining option, yielding $(b,a)$.
The TCTC outcome $(a,b)$ Pareto dominates this allocation: student 1 is indifferent and student 2 strictly prefers $b$.

\begin{figure}[htbp]
\centering
\definecolor{cycleblue}{HTML}{0072B2}
\definecolor{fixedteal}{HTML}{007F6D}
\begin{tikzpicture}[
  font=\small,
  student/.style={circle,draw=cycleblue!75!black,fill=white,line width=0.6pt,minimum size=7mm,inner sep=1pt},
  processed/.style={student,dashed,draw=fixedteal,text=fixedteal,fill=fixedteal!7},
  classframe/.style={draw=fixedteal!60,rounded corners=4pt,line width=0.5pt},
  exchange/.style={-{Stealth[length=2.4mm]},draw=cycleblue,line width=0.9pt},
  school/.style={draw=black!35,fill=black!2,rounded corners=3pt,line width=0.5pt,minimum width=17mm,minimum height=24mm},
  heading/.style={font=\small\bfseries}
]
\node[heading] at (1,2.05) {Initial endowment};
\node[school] at (0,0) {};
\node[school] at (2,0) {};
\node at (0,1.48) {School $a$};
\node at (2,1.48) {School $b$};
\node[student] (s1) at (0,0.55) {$1$};
\node[student] (s2) at (0,-0.55) {$2$};
\draw[-{Stealth[length=2mm]},draw=black!50,line width=0.6pt] (s2.north) -- (s1.south);
\draw[exchange] (s1.south west) .. controls (-1.05,-0.08) and (-1.05,1.18) .. (s1.north west);
\node at (1,-1.58) {School sizes: $(2,0)$};
\draw[-{Stealth[length=2.4mm]},draw=black!60,line width=0.8pt]
  (3.15,0) -- node[above,text=black!75] {Fix $T_1$} (4.15,0);

\node[heading] at (6.3,2.05) {After fixing $T_1$};
\node[school] at (5.3,0) {};
\node[school] at (7.3,0) {};
\node at (5.3,1.48) {School $a$};
\node at (7.3,1.48) {School $b$};
\draw[classframe] (4.85,0.10) rectangle (7.75,1.00);
\node[processed] at (5.3,0.55) {$1$};
\node[processed] at (7.3,0.55) {$1$};
\node[student] (t2) at (5.3,-0.55) {$2$};
\draw[exchange] (t2.south west) .. controls (4.25,-1.18) and (4.25,0.08) .. (t2.north west);
\node[text=fixedteal,font=\footnotesize] at (6.3,-1.58) {Dashed circles: fixed class};
\draw[-{Stealth[length=2.4mm]},draw=black!60,line width=0.8pt]
  (8.50,0) -- node[above,align=center,text=black!75] {Fix $T_2$\\and allocate} (9.80,0);

\node[heading] at (12,2.05) {Final allocation};
\node[school] at (11,0) {};
\node[school] at (13,0) {};
\node at (11,1.48) {School $a$};
\node at (13,1.48) {School $b$};
\node[student] at (11,0.55) {$1$};
\node[student] at (13,-0.55) {$2$};
\node at (12,-1.58) {School sizes: $(1,1)$};
\end{tikzpicture}
\caption{TCTC assigns a student to an initially empty school.}
\label{fig:tctc-example}
\end{figure}
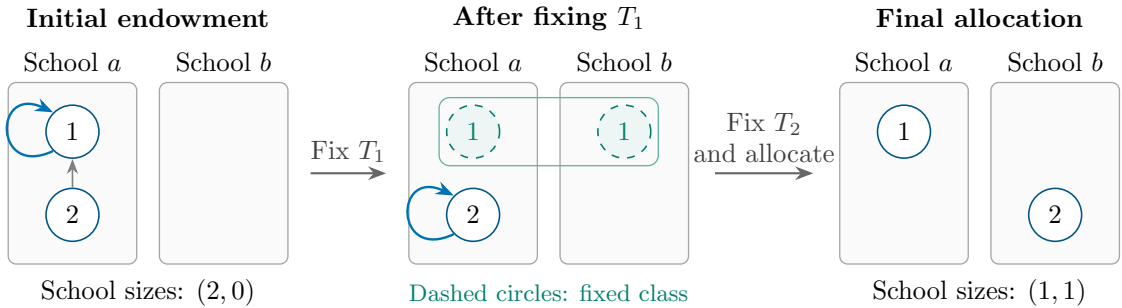
\end{example}
\FloatBarrier

\section{Main results}\label{sec:analysis}

Our main result is the following theorem for base markets.

\begin{theorem}\label{thm:base-market}
In every base market, TCTC is individually rational, Pareto efficient, and weakly group-strategy-proof on the full domain of weak preferences.
Given a membership oracle for $\B$, TCTC runs in time polynomial in $n$ and $|E|$.
\end{theorem}

We prove order independence of TCTC's final allocation in \Cref{prop:order}.
We establish IR and PE in \Cref{prop:irpe}, WGSP in \Cref{prop:wgsp}, and polynomial-time computability in \Cref{prop:complexity}.

On the strict-preference domain, TCTC can be viewed as the TTC-M mechanism of \citet{SuzukiEtAl2023}.
To see this, regard each object as a school, with feasible occupancy vectors given by the incidence vectors of the bases in $\B$.
These vectors form an M-convex set, since the M-convex subsets of $\{0,1\}^E$ are exactly the sets of incidence vectors of matroid base families on $E$.
Extend each agent $i$'s strict preference to $E$ by placing $E\setminus E_i$ below $E_i$ in a fixed order.
Use the same common agent priority as TCTC, giving each school's initial holder, if any, first priority.

Under strict preferences, fixed classes are singletons, and \eqref{eq:residual} gives exactly the admissibility tests of TTC-M with past assignments fixed.
Each agent's initial school remains available until she is processed, and hence she always points to a school in $E_i$.
Each school points to the highest-priority admissible agent, and suppressing the school vertices gives the pointing graph of TCTC.
Thus, the mechanisms produce the same allocation, and the strong group-strategy-proofness of TTC-M \citep[Theorem~6]{SuzukiEtAl2023} gives the following corollary.

\begin{corollary}\label{cor:strict-gsp}
In every base market, TCTC is strongly group-strategy-proof when both preferences and admissible reports are restricted to be strict.
\end{corollary}

The class of g-matroid school constraints is maximal in the following sense.
\Citet[Theorem~5]{ImamuraKawase2025} fix a student set $I$, a school set $S$ with $|S|\ge3$, and a school $s^*$ whose nonempty constraint family $\cF_{s^*}$ is not a g-matroid.
They show that there is a feasible initial endowment such that no mechanism satisfies IR, PE, and SP on the strict-preference domain when every other school has a unit-capacity constraint.
The strict-preference domain is contained in the full domain of weak preferences, and WGSP implies SP.
Hence, adding any such family to the class of permissible school constraints rules out a uniform guarantee of IR, PE, and SP, and \Cref{thm:base-market} gives the following characterization.

\begin{corollary}\label{cor:maximality}
Within the class of school-by-school constraints of \Cref{subsec:model}, the class of g-matroids is maximal among classes of nonempty school constraint families for which every market with a feasible initial endowment admits a mechanism satisfying IR, PE, and SP on the full domain of weak preferences.
\end{corollary}

In \Cref{subsec:preparation}, we establish two properties of matroid exchanges and use them to prove feasibility of cycle processing and order independence.
\Cref{subsec:welfare} proves IR, PE, and WGSP from a single domination lemma, and \Cref{sec:computation} gives the polynomial-time implementation.

\subsection{Structure of cycle processing}\label{subsec:preparation}

We use the following standard notions for a matroid base family $\B$ on a ground set $U$.
For $A\subseteq U$, the rank $r_\B(A)$ is the largest number of elements in an independent subset of $A$.
The closure of $A$ is $\cl_\B(A)=\{e\in U\mid r_\B(A+e)=r_\B(A)\}$.
We call $\cl_\B(A)$ the flat spanned by $A$.
An element is a loop if and only if it belongs to $\cl_\B(\emptyset)$.
A circuit is a minimal set that is not independent.
If $B$ is a base and $e\notin B$ is a nonloop, then $B+e$ contains a unique circuit $C_B(e)$, the fundamental circuit of $e$ with respect to $B$, and for $b\in B$ the set $B-b+e$ is a base if and only if $b\in C_B(e)-e$.

Fix a state with processed set $F$ and residual base family $\tilde{\B}$, and let $r=|\bar F|$.
We assume that $\omega(\bar F)\in\tilde{\B}$, which holds initially and is preserved by cycle processing, as shown in \Cref{prop:cyclefeasible} below.
For each $i\in\bar{F}$, let $T_i$ and $\tau(i)$ denote the top class and the pointing target computed at this state.
For a nonloop $a$ of $\tilde{\B}$, let $\pi(a)$ be the smallest index $j\in\bar F$ such that $\omega(\bar F)-\omega(j)+a\in\tilde{\B}$.
We now show that such an index exists.
If $a\notin \omega(\bar F)$, the initial endowments that $a$ can replace are the elements of the fundamental circuit $C_{\omega(\bar F)}(a)-a$, which is nonempty because $a$ is a nonloop.
If $a=\omega(j')\in \omega(\bar F)$, then $\omega(\bar F)-\omega(j)+a\in\tilde\B$ only for $j=j'$, and $\pi(a)=j'$.
Thus, $\omega(\pi(a))$ is the highest-priority initial endowment that $a$ can replace in $\omega(\bar F)$, and \eqref{eq:target} can be rewritten as
\begin{align}
  \tau(i)=\min\{\pi(a)\mid a\in T_i\text{ is a nonloop of }\tilde{\B}\}
  \qquad(i\in\bar F).
  \label{eq:tau-pi}
\end{align}
Let $\sigma(1),\ldots,\sigma(r)$ be the remaining agents in priority order.
The following lemma is used to establish feasibility of cycle processing and order independence.

\begin{lemma}\label{lem:priority}
If $Y\subseteq E_{\bar F}$ is a set of nonloops of $\tilde\B$ and the restriction $\pi|_Y\colon Y\to\bar F$ is injective, then $(\omega(\bar F)\setminus\{\omega(\pi(y))\mid y\in Y\})\cup Y\in\tilde{\B}$.
Moreover, for every $Z\in\tilde\B$ and $\ell\in\{0,1,\ldots,r\}$, at least $\ell$ elements $z\in Z$ satisfy $\pi(z)\in\{\sigma(1),\ldots,\sigma(\ell)\}$.
\end{lemma}

\begin{proof}
For each $\ell\in\{1,\ldots,r\}$, define $S_\ell=\{\omega(\sigma(\ell)),\ldots,\omega(\sigma(r))\}$.
Also, let $S_{r+1}=\emptyset$.
Each $S_\ell$ is independent and has rank $r-\ell+1$.
We first show that, for every nonloop $a$ and every $\ell\in\{1,\ldots,r\}$,
\begin{equation}
  \pi(a)=\sigma(\ell)
  \quad\Longleftrightarrow\quad
  a\in\cl_{\tilde\B}(S_\ell)\setminus\cl_{\tilde\B}(S_{\ell+1}).
  \label{eq:priority-closure}
\end{equation}
If $a\notin \omega(\bar F)$, then $a$ lies in the closure of a subset $T$ of $\omega(\bar F)$ if and only if $C_{\omega(\bar F)}(a)-a\subseteq T$.
Hence, the highest-priority element of $C_{\omega(\bar F)}(a)-a$ is $\omega(\sigma(\ell))$ if and only if $S_\ell$ spans $a$ and $S_{\ell+1}$ does not.
If $a=\omega(\sigma(\ell))\in \omega(\bar F)$, then $\pi(a)=\sigma(\ell)$, and the independence of $\omega(\bar F)$ gives $a\in\cl_{\tilde\B}(S_\ell)\setminus\cl_{\tilde\B}(S_{\ell+1})$.

For the first assertion, let $y_\ell$ be the element of $Y$ with $\pi(y_\ell)=\sigma(\ell)$ if it exists, and $y_\ell=\omega(\sigma(\ell))$ otherwise.
Then, we have $(\omega(\bar F)\setminus\{\omega(\pi(y))\mid y\in Y\})\cup Y=\{y_1,\ldots,y_r\}$.
Moreover, $y_\ell\in\cl_{\tilde\B}(S_\ell)\setminus\cl_{\tilde\B}(S_{\ell+1})$ for every $\ell$.
For each $\ell$, all of $y_{\ell+1},\ldots,y_r$ lie in $\cl_{\tilde\B}(S_{\ell+1})$, whereas $y_\ell$ does not.
Thus, adding $y_r,\ldots,y_1$ in this order increases the rank at every step, and the resulting set of $r$ elements is a base of $\tilde\B$.

For the second assertion, the closure $\cl_{\tilde\B}(S_{\ell+1})$ has rank $r-\ell$ and hence contains at most $r-\ell$ elements of the independent set $Z$.
Since $|Z|=r$, at least $\ell$ elements of $Z$ lie outside this closure.
By \eqref{eq:priority-closure}, these are exactly the elements $z\in Z$ with $\pi(z)\in\{\sigma(1),\ldots,\sigma(\ell)\}$.
\end{proof}

We now use the first assertion of \Cref{lem:priority} to show that a selected cycle can be processed and that the updated residual base family has the properties needed in later iterations.

\begin{proposition}\label{prop:cyclefeasible}
Let $C$ be a selected cycle, and let $\tilde{\B}'$ be the residual base family defined by \eqref{eq:residual} for the processed set $F\cup C$.
Then $\omega(\bar F\setminus C)$ belongs to $\tilde{\B}'$.
The family $\tilde{\B}'$ is a rank-$|\bar{F}\setminus C|$ matroid base family, and each of its nonloops is a nonloop of $\tilde{\B}$.
\end{proposition}

\begin{proof}
For each $k\in C$, \eqref{eq:tau-pi} provides a nonloop $x_k\in T_k$ with $\pi(x_k)=\tau(k)$.
The objects $x_k$ are distinct because the sets $E_k$ are pairwise disjoint.
The values $\pi(x_k)=\tau(k)$ are distinct and $\{\tau(k)\mid k\in C\}=C$ because $C$ is a directed cycle.
The first assertion of \Cref{lem:priority} gives $\omega(\bar F\setminus C)\cup\{x_k\mid k\in C\}\in\tilde{\B}$, and the definition of the residual base family yields $\omega(\bar F\setminus C)\in\tilde{\B}'$.

In particular, $\tilde{\B}'$ is nonempty.
By expanding \eqref{eq:residual}, a set belongs to $\tilde{\B}'$ if and only if it can be combined with one object from each class $T_k$, for $k\in C$, to form a base of $\tilde{\B}$.
Applying \Cref{lem:residual} to $\tilde{\B}$ with processed set $C$ and fixed classes $(T_k)_{k\in C}$ gives the remaining claims.
\end{proof}

The following persistence property is the key to order independence.

\begin{lemma}\label{lem:persistence}
Let $C$ be a selected cycle, and suppose that $i\notin C$ points to $j=\tau(i)\notin C$.
After the mechanism processes $C$, agent $i$ has the same top class $T_i$ and still points to $j$.
\end{lemma}

\begin{proof}
Let $\tilde{\B}'$ be the updated residual base family.

We first show that $T_i$ remains the top class.
By \eqref{eq:tau-pi}, there are a nonloop $x_i\in T_i$ with $\pi(x_i)=j$ and, for each $k\in C$, a nonloop $x_k\in T_k$ with $\pi(x_k)=\tau(k)$.
These values are pairwise distinct because $j\notin C=\{\tau(k)\mid k\in C\}$.
The first assertion of \Cref{lem:priority} gives $(\omega(\bar F\setminus C)-\omega(j))\cup\{x_k\mid k\in C\}\cup\{x_i\}\in\tilde\B$, and \eqref{eq:residual} yields $\omega(\bar F\setminus C)-\omega(j)+x_i\in\tilde{\B}'$.
Thus, $T_i$ contains a nonloop of $\tilde{\B}'$.
By \Cref{prop:cyclefeasible}, the update creates no new nonloop.
Therefore, no class that agent $i$ strictly prefers to $T_i$ contains a nonloop of $\tilde{\B}'$, and $T_i$ remains her top class.

Since $j$ remains available, it suffices to rule out a target $h<j$ after processing $C$.
Suppose that such a target $h\in\bar F\setminus C$ is available to $i$.
By \eqref{eq:residual}, there are an object $a\in T_i$ and objects $z_k\in T_k$ for $k\in C$ such that
\begin{align*}
  (\omega(\bar F\setminus C)-\omega(h)+a)\cup\{z_k\mid k\in C\}\in\tilde{\B}.
\end{align*}
Let $P=\{k\in\bar F\mid k\le h\}$.
By the second assertion of \Cref{lem:priority} with $\ell=|P|$, at least $|P|$ elements $z$ of this base must satisfy $\pi(z)\in P$.
Since $\pi(\omega(k))=k$ for every $k\in\bar F$, we have
\begin{align*}
  \{z\in\omega(\bar F\setminus C)-\omega(h)\mid\pi(z)\in P\}
  =\{\omega(k)\mid k\in P\setminus(C\cup\{h\})\}.
\end{align*}
This set has $|P|-|C\cap P|-1$ elements because $h\in P\setminus C$.
By \eqref{eq:tau-pi}, $\pi(a)\ge j>h$, and hence $\pi(a)\notin P$.
Also by \eqref{eq:tau-pi}, we have $\pi(z_k)\ge\tau(k)$ for every $k\in C$.
Since $\tau$ permutes $C$, at most $|C\cap P|$ of the objects $z_k$ satisfy $\pi(z_k)\in P$.
Consequently, at most $(|P|-|C\cap P|-1)+|C\cap P|=|P|-1$ elements $z$ of this base satisfy $\pi(z)\in P$, contradicting the required lower bound.
\end{proof}

Processing one cycle leaves every other cycle available with the same top classes, by \Cref{lem:persistence}.
This preservation property yields order independence.

\begin{proposition}\label{prop:order}
For every reported profile, all orders of processing selected cycles fix the same class $T_i$ for each agent $i$.
Moreover, some feasible allocation $\mu$ satisfies $\mu(i)\in T_i$ for every agent $i$.
Consequently, the final allocation of TCTC does not depend on which cycle is selected at each step.
\end{proposition}

\begin{proof}
We prove that all executions from any reachable state fix the same classes, by induction on the number of remaining agents.
The claim is immediate when no agent remains.
Consider two executions from the same state, and let $C$ and $D$ be their first cycles.
If $C=D$, the two executions fix the same classes by the induction hypothesis applied after this cycle.
Otherwise, $C$ and $D$ are disjoint, and \Cref{lem:persistence} shows that each remains available with the same classes and pointing targets after the other is processed.
After $C$ is processed, fewer agents remain, and the induction hypothesis implies that all continuations fix the same classes.
We may therefore choose a continuation that processes $D$ next without changing the classes fixed by the first execution.
Similarly, the induction hypothesis after $D$ allows us to choose $C$ next in the second execution.
After these two steps, both executions have the same processed set and fixed classes.
By \eqref{eq:residual}, they also have the same residual base family, and the remaining agents have the same reports.
The executions can therefore follow the same sequence of cycles and fix the same classes for all remaining agents.
Thus, the two original executions fix the same classes for all agents.

When the mechanism terminates, every agent's class has been fixed and $F=I$.
At this point, \Cref{prop:cyclefeasible} and \eqref{eq:residual} guarantee a feasible allocation $\mu$ with $\mu(i)\in T_i$ for every $i$.
Since the final selection depends only on $(T_i)_{i\in I}$, the outcome is also independent of the cycle order.
\end{proof}

\subsection{Welfare and incentives}\label{subsec:welfare}

We now prove that TCTC satisfies IR, PE, and WGSP, and establish a core property.
We begin with a lemma that will be used in the proofs of both PE and WGSP.

\begin{lemma}\label{lem:domination}
Fix a reported profile, and let $\mu$ be the outcome of TCTC.
Let $F$ be the set of agents processed before the cycle containing agent $i$, under any order of cycle processing.
If a feasible allocation $\nu$ satisfies $\nu(f)\in T_f$ for every $f\in F$, then $\mu(i)\succsim_i\nu(i)$.
\end{lemma}

\begin{proof}
Let $\tilde\B$ be the residual base family immediately before this cycle is processed.
Because $\nu\in\B$ and $\nu(f)\in T_f$ for every $f\in F$, \eqref{eq:residual} shows that $\{\nu(k)\mid k\in\bar F\}$ is a base of $\tilde{\B}$.
Hence, $\nu(i)$ is a nonloop of $\tilde{\B}$.
Processing the cycle containing $i$ fixes $T_i$, her most preferred class that contains a nonloop of $\tilde\B$, and therefore every object in $T_i$ is weakly preferred by $i$ to $\nu(i)$.
By \Cref{prop:order}, $\mu(i)\in T_i$.
\end{proof}

IR follows because each agent's initial endowment remains available until she is processed.
We use \Cref{lem:domination} to rule out Pareto improvements.

\begin{proposition}\label{prop:irpe}
TCTC is individually rational and Pareto efficient for the reported weak preferences.
\end{proposition}

\begin{proof}
Let $\mu$ be the outcome.
Immediately before agent $i$ is processed, the base $\omega(\bar{F})$ of the current residual base family contains $\omega(i)$, and hence $\omega(i)$ is a nonloop of that family.
The mechanism fixes agent $i$'s best class among all such nonloops and eventually assigns her an object in that class.
Therefore, $\mu(i)\succsim_i\omega(i)$.

For PE, suppose that a feasible allocation $\nu$ Pareto dominates $\mu$.
Fix any cycle order and consider the first processed cycle that contains an agent $i$ with $\nu(i)\succ_i\mu(i)$.
Let $F$ be the agents processed before this cycle.
For every $f\in F$, Pareto dominance and the choice of the cycle imply $\nu(f)\sim_f\mu(f)$, and since $T_f$ is the entire indifference class of $\mu(f)$, we have $\nu(f)\in T_f$.
\Cref{lem:domination} gives $\mu(i)\succsim_i\nu(i)$, a contradiction.
\end{proof}

Order independence lets the truthful and deviating executions begin with the same sequence of cycles that contain no coalition agent.

\begin{proposition}\label{prop:wgsp}
TCTC is weakly group-strategy-proof.
\end{proposition}

\begin{proof}
Fix a truthful profile $\succsim_I$, a coalition $Q$, and reports $\succsim'_Q$.
Let $\mu=\Phi(\succsim_I)$ and $\mu'=\Phi(\succsim'_Q,\succsim_{-Q})$, and suppose, contrary to WGSP, that $\mu'(i)\succ_i\mu(i)$ for every $i\in Q$.

Run the truthful execution as follows: as long as the pointing graph contains a cycle with no agent in $Q$, process such a cycle.
Let $F\subseteq I\setminus Q$ be the set of agents processed when this stops.
The deviating execution can process the same sequence of cycles.
Indeed, at each step the two executions have the same processed set and the same residual base family, and the top class and pointing target of an agent outside $Q$ depend only on her unchanged report and this state.
Hence, the same cycle with no agent in $Q$ is present in both pointing graphs, and processing it fixes the same classes in both executions.
By \Cref{prop:order}, the classes fixed for the agents in $F$ are their final classes under both profiles, and in particular $\mu'(f)\in T_f$ for every $f\in F$.

At the state reached by the truthful execution, every cycle of the pointing graph contains an agent in $Q$.
Choose such a cycle and an agent $i\in Q$ on it.
\Cref{lem:domination} applied to $\nu=\mu'$ gives $\mu(i)\succsim_i\mu'(i)$, contradicting $\mu'(i)\succ_i\mu(i)$.
\end{proof}

\Cref{lem:domination} also rules out coalitions that make every member strictly better off through an independent partial allocation whose objects span the same flat as their initial endowments in the original matroid.

\begin{corollary}\label{cor:core}
Let $\mu$ be the outcome of TCTC.
Let $Q\subseteq I$ be a nonempty coalition and let $\nu\colon Q\to E$ be a partial allocation with $\nu(i)\in E_i$ for every $i\in Q$.
If $\{\nu(i)\mid i\in Q\}$ is independent and
\begin{align*}
  \cl_{\B}(\{\nu(i)\mid i\in Q\})=\cl_{\B}(\omega(Q)),
\end{align*}
then there is an $i\in Q$ with $\mu(i)\succsim_i\nu(i)$.
\end{corollary}

\begin{proof}
Fix $Q$ and $\nu$ as in the statement.
Consider the first processed cycle containing an agent in $Q$, and let $F\subseteq I\setminus Q$ be the set of agents processed before this cycle.
By \Cref{prop:cyclefeasible} and \eqref{eq:residual}, there is a base $B\in\B$ containing $\omega(\bar F)$ and one object from each fixed class $T_f$ for $f\in F$.
Define
\begin{align*}
  \widehat\nu=(B\setminus\omega(Q))\cup\{\nu(i)\mid i\in Q\}.
\end{align*}
Since $\{\nu(i)\mid i\in Q\}$ and $\omega(Q)$ span the same flat, $\widehat\nu$ spans $B$.
It has $n$ elements and hence is a base of $\B$.
It contains exactly one object from every $E_i$ and satisfies $\widehat\nu(f)\in T_f$ for every $f\in F$.
Thus, $\widehat\nu$ is a feasible allocation to which \Cref{lem:domination} applies.
For an agent $i\in Q$ on the selected cycle, that lemma gives $\mu(i)\succsim_i\widehat\nu(i)=\nu(i)$.
\end{proof}

\Cref{cor:core} strengthens individual rationality, which follows by taking $Q=\{i\}$ and $\nu(i)=\omega(i)$.
For a general coalition $Q$, the partial allocation $\nu$ can be extended to a feasible allocation while leaving every agent outside $Q$ with her initial endowment.
Indeed, since the objects assigned to $Q$ span the same flat as $\omega(Q)$, the set $\omega(I\setminus Q)\cup\{\nu(i)\mid i\in Q\}$ is a base of $\B$ and contains exactly one object from each $E_i$.

In the standard housing market with $n$ agents and $n$ houses, represented by a partition matroid with one capacity-one block per house, $\cl_{\B}(\omega(Q))$ consists of the assignment options for houses initially owned by members of $Q$.
Thus, \Cref{cor:core} specializes to the usual core condition that rules out strict improvement for every member of a coalition~\citep{ShapleyScarf1974}.

\subsection{Computation}\label{sec:computation}

We show that TCTC can be implemented in polynomial time given a membership oracle for $\B$.

\begin{proposition}\label{prop:complexity}
Given a membership oracle for $\B$, TCTC can be implemented in time polynomial in $n$ and $|E|$, counting each oracle call as one operation.
\end{proposition}

\begin{proof}
The mechanism performs at most $n$ iterations because each iteration fixes the class of at least one remaining agent.
It suffices to show that the top classes, pointing targets, and final allocation can be computed in polynomial time.
Fix a state with processed set $F$ and a set $A\subseteq E_{\bar F}$ of size $|\bar F|$.
Consider the partition matroid whose bases consist of $A$ together with one object from each fixed class $T_f$.
By \eqref{eq:residual}, $A\in\tilde\B$ if and only if this matroid and $\B$ have a common base.
The common-base problem is solvable in polynomial time using the membership oracle for $\B$ and the known base $\omega(I)$ (\citealp[Theorem~41.4]{Schrijver2003} and \citealp[p.~175]{LovaszRecski1982}).
Testing whether $\omega(\bar F)-\omega(j)+a\in\tilde\B$ for every $j\in\bar F$ and $a\in E_{\bar F}$ determines the nonloops, top classes, and pointing targets, with candidates of smaller cardinality rejected immediately.
There are at most $|\bar F|\cdot|E_{\bar F}|$ such tests per iteration, giving $O(n^2|E|)$ tests in total.
Finding a cycle and updating the state also take polynomial time.

For the final selection, feasible completions are the common bases of $\B$ and the partition matroid with capacity-one blocks $T_i$ for all $i\in I$ and all other objects treated as loops.
Enumerate the objects as $e_1\sqsubset\cdots\sqsubset e_m$, where $m=|E|$, and assign weight $-2^{m-k}$ to $e_k$.
The minimum-weight common base is exactly the minimum feasible completion with respect to $\triangleleft$.
Since each weight has $O(m)$ bits, finding a minimum-weight common base in these two matroids takes polynomial time~\citep[Theorem~41.8]{Schrijver2003}.
\end{proof}

\section*{Acknowledgments}
This work was supported by JSPS KAKENHI Grant Number JP25K00137 and JST ERATO Grant Number JPMJER2301.

The author used AI tools, including OpenAI's Codex, to assist with developing research ideas, constructing and checking proofs, and drafting and revising the manuscript.
The author takes full responsibility for the content of this paper, including all mathematical claims and proofs.

\bibliographystyle{plainnat}
\bibliography{references}

\end{document}